\documentclass[final,copyright,creativecommons]{eptcs}
\providecommand{\event}{DCM 2023} 

\usepackage{iftex}

\ifpdf
  \usepackage{underscore}         
  \usepackage[T1]{fontenc}        
\else
  \usepackage{breakurl}           
\fi

\usepackage{amsmath}
\usepackage{amssymb}
\usepackage{amsthm}

\theoremstyle{plain}
\newtheorem{theorem}{Theorem}
\newtheorem{proposition}{Proposition}
\newtheorem{lemma}{Lemma}
\newtheorem{corollary}{Corollary}

\theoremstyle{definition}
\newtheorem{definition}{Definition}

\theoremstyle{remark}

\usepackage{newtxtext}
\usepackage{newtxmath}

\usepackage{etoolbox}

\usepackage{stmaryrd}
\usepackage[inline]{enumitem}

\usepackage[retainorgcmds]{IEEEtrantools}

\usepackage[textsize=scriptsize, textwidth=2cm, obeyFinal]{todonotes}

\usepackage{tikz}
\usetikzlibrary{positioning,automata,arrows,hobby,decorations.pathreplacing,fit}
\tikzstyle{every state}=[minimum size=1.5em]
\tikzset{>=latex,auto,node distance=2.5cm, every
  loop/.style={looseness=6}, initial text={}, inner sep=0.5mm,
  loopright/.style={loop,looseness=6,out=35, in=-35},
  loopleft/.style={loop,looseness=6,out=145, in=215},
  loopabove/.style={loop,looseness=6,out=125, in=55},
  loopbelow/.style={loop,looseness=6,out=-125, in=-55}, }

\usepackage[nocompress]{cite}

\usepackage{microtype}
\AtBeginEnvironment{verbatim}{\microtypesetup{activate=false}}

\usepackage{hyperref}
\usepackage{caption}

\makeatletter
\newcommand{\oset}[3][0ex]{%
  \mathrel{\mathop{#3}\limits^{
    \vbox to#1{\kern-2\ex@
    \hbox{$\scriptstyle#2$}\vss}}}}
\makeatother

\DeclareMathOperator{\eventually}{\mathbf{F}}
\DeclareMathOperator{\weventually}{\overline{\mathbf{F}}}
\DeclareMathOperator{\once}{\overleftarrow{\mathbf{F}}}
\DeclareMathOperator{\globally}{\mathbf{G}}

\DeclareMathOperator{\nextx}{\mathbf{X}}

\DeclareMathOperator{\counting}{\mathbf{C}}
\DeclareMathOperator{\pcounting}{\vphantom{\mathbf{C}}\smash{\overleftarrow{\mathbf{C}}}}
\DeclareMathOperator{\pnueli}{\mathbf{P}}
\DeclareMathOperator{\ppnueli}{\vphantom{\mathbf{P}}\smash{\overleftarrow{\mathbf{P}}}}
\DeclareMathOperator{\until}{\mathbin{\mathbf{U}}}

\DeclareMathOperator{\since}{\mathbin{\mathbf{S}}}
\DeclareMathOperator{\release}{\mathbin{\mathbf{R}}}

\newcommand*{\ltl}{\textup{\textmd{\textsf{LTL}}}}

\newcommand*{\onetptl}{\textup{\textmd{\textsf{1-TPTL[$\until{}, \since{}$]}}}}

\newcommand*{\emitl}{\textup{\textmd{\textsf{EMITL}}}}
\newcommand*{\pnemtl}{\textup{\textmd{\textsf{PnEMTL}}}}

\newcommand*{\mitl}{\textup{\textmd{\textsf{MITL}}}}

\newcommand*{\mtl}{\textup{\textmd{\textsf{MTL}}}}
\newcommand*{\msoone}{\textup{\textmd{\textsf{MSO[$<, +1$]}}}}
\newcommand*{\foone}{\textup{\textmd{\textsf{FO[$<, +1$]}}}}
\newcommand*{\mso}{\textup{\textmd{\textsf{MSO[$<$]}}}}
\newcommand*{\fo}{\textup{\textmd{\textsf{FO[$<$]}}}}
\newcommand*{\qtwomlo}{\textup{\textmd{\textsf{Q2MLO}}}}
\newcommand*{\qtwomso}{\textup{\textmd{\textsf{Q2MSO}}}}

\newcommand*{\tlc}{\textup{\textmd{\textsf{TLC}}}}
\newcommand*{\tlci}{\textup{\textmd{\textsf{TLCI}}}}
\newcommand*{\tlp}{\textup{\textmd{\textsf{TLP}}}}
\newcommand*{\tlpi}{\textup{\textmd{\textsf{TLPI}}}}

\newcommand*{\nfa}{\textup{\textmd{\textsf{NFA}}}}
\newcommand*{\dfa}{\textup{\textmd{\textsf{DFA}}}}

\newcommand*\A{\mathcal A} 
\newcommand*\B{\mathcal B} 

\newcommand*\transitions{\Delta}

\newcommand*\R{\mathbb R}

\newcommand*\N{\mathbb N}

\renewcommand*\phi{\varphi}
\newcommand*\AP{\textup{\textmd{\textsf{AP}}}}

\renewcommand*{\figurename}{Fig.}

\title{When Do You Start Counting?\\ \large Revisiting Counting and Pnueli Modalities in Timed Logics}

\author{Hsi-Ming~Ho
\institute{Department of Informatics, University of Sussex\\
United Kingdom}
\email{hsi-ming.ho@sussex.ac.uk}
\and
Khushraj~Madnani
\institute{Max Planck Institute for Software Systems\\
Germany}
\email{kmadnani@mpi-sws.org}
}

\def\titlerunning{When Do You Start Counting? Revisiting Counting and Pnueli Modalities in Timed Logics}
\def\authorrunning{H.-M. Ho \& K. Madnani}
\begin{document}
\maketitle

\begin{abstract}
Pnueli first noticed that certain simple `counting' properties appear to be inexpressible in popular timed temporal logics such as
Metric Interval Temporal Logic (\mitl{}). 
This interesting observation has since been studied extensively, culminating in strong timed logics that are capable of expressing such properties yet remain decidable.
A slightly more general case, namely where one asserts the existence of a sequence of events in an arbitrary interval of the form $\langle a, b \rangle$ (instead of an upper-bound interval of the form $[0, b \rangle$, which starts from the current point in time),
has however not been addressed satisfactorily in the existing literature.
We show that counting in $[0, b \rangle$ is in fact as powerful as counting in $\langle a, b \rangle$; moreover, the general property `there exist $x', x'' \in I$ such that $x' \leq x''$ and $\psi(x', x'')$ holds' can be expressed in Extended Metric Interval Temporal Logic (\emitl{}) with only $[0, b \rangle$.

\end{abstract}

\section{Introduction}
\paragraph{Timed logics.}
Temporal logics provide constructs to specify \emph{qualitative} ordering between events in time.
Timed logics extend classical temporal logics with the ability to specify \emph{quantitative} timing constraints between events. \emph{Metric Interval Temporal Logic} (\mitl{})~\cite{AluFed96} is amongst the best studied of timed logics. It extends the `until' ($\until$) and `since' ($\since$) modalities of \emph{Linear Temporal Logic} (\ltl{})~\cite{Pnueli1977} with \emph{non-singular} intervals to specify timing constraints. For example, $P \until_I Q$ states that an event where $Q$ holds should occur in the future within a time interval $I$, and $P$ should hold continuously till then.

\paragraph{Specifying multiple events.}
In many practical scenarios, e.g,~those involving resource-bounded computations, the ability to specify not just one but a sequence of events within a given time interval can be crucial. 
 For example, in a multi-threaded environment, a desired property for scheduling algorithms could be to have at most $k$ context switches in every $M$ time units. Such properties, however, cannot be expressed in \mitl{}~\cite{Bouyer2010, Hirshfeld06, KriMad16}.
In particular, the \emph{counting} ($\counting$ and $\pcounting$) and \emph{Pnueli} ($\pnueli$ and $\ppnueli$) modalities that specify event occurrences within the \emph{next} or \emph{previous} unit interval (i.e.~within $[t_0, t_0+1)$ or $(t_0-1, t_0]$, where the current time is $t_0$) are studied in~\cite{Hirshfeld06}, and it turned out that for \mitl{} extended with these modalities (called \tlc{} and \tlp{}, respectively),
the satisfiability problem remains $\mathrm{EXPSPACE}$-complete.\footnote{The exponential blow-up comes from  the succinct encodings of both constants in intervals of the form $\langle a, b \rangle$ in \mitl{} and constants $k$ in $\counting_I^{k}$; for more details, see~\cite{Rabinovich2010}.} Moreover, it turned out that
\tlc{} and \tlp{}, while the latter is syntactically more general, are equally expressive in the continuous semantics. This is shown by proving that both \tlc{} and \tlp{} are expressively complete for a natural fragment of \emph{Monadic First-Order Logic of Order and Metric} (\foone{}) called \qtwomlo{}, where one can specify that
the sequence of events between the current time $t_0$ and $t \in t_0 + I$ (for a \emph{non-singular} interval $I$) satisfies a first-order formula $\vartheta(x_0, x)$.

\begin{figure}[t]
\centering
\begin{IEEEeqnarray*}{rCl}
   \ltl{} & = & \textsf{Propositional Logic} \cup  \{ \phi_1 \until \phi_2, \phi_1 \since \phi_2  \, \mid \, \phi_1, \phi_2 \in \ltl{} \} \\
   \mitl{} & = & \ltl{} \cup \{ \phi_1 \until_I \phi_2, \phi_1 \since_I \phi_2 \, \mid \, \phi_1, \phi_2 \in \mitl{} \text{, } I = \langle a, b \rangle \text{, } a, b \in \N \cup\{\infty\} \text{, } a < b  \} \\
   \tlc{} & = & \mitl + \{ \counting_I^{k} \varphi, \pcounting_I^{k} \varphi \, \mid \, \varphi \in \tlc{} \text{, } I = [0, b \rangle \text{, } b \geq 1 \text{, } k \geq 1 \} \\
   \tlp{} & = & \mitl + \{ \pnueli^{k}_I \varphi, \ppnueli^{k}_I \varphi \, \mid \, \varphi \in \tlp{} \text{, } I = [0, b \rangle \text{, } b \geq 1 \text{, } k \geq 1 \} \\
   \tlci{} & = & \mitl + \{ \counting_I^{k} \varphi, \pcounting_I^{k} \varphi \, \mid \, \varphi \in \tlci{} \text{, } I = \langle a, b \rangle \text{, } a, b \in \N \cup\{\infty\} \text{, } a < b \} \\
   \tlpi{} & = & \mitl + \{ \pnueli^{k}_I \varphi, \ppnueli^{k}_I \varphi \, \mid \, \varphi \in \tlpi{} \text{, } I = \langle a, b \rangle \text{, } a, b \in \N \cup\{\infty\} \text{, } a < b \}
\end{IEEEeqnarray*}
\caption{Some timed temporal logics considered in this paper. Note that the definitions of \tlc{} and \tlp{} in~\cite{Hirshfeld06} are less general but equally expressive in the continuous semantics.}
\label{fig:zoo}
\end{figure}

\paragraph{Expressiveness.} 
It is of course trivial to see that \qtwomlo{} subsumes $\tlc{}$, but it is unclear (at least to us) whether \qtwomlo{} can express the more general modalities $\counting^k_I$ and their past counterparts, which count event occurrences within arbitrary non-singular intervals $I$ of the form $\langle a, b \rangle$ with $0 \leq a < b$---on the face of it, we seem to need a first-order formula $\vartheta(x_1, x_2)$ along with two quantified instants $t_1, t_2 \in t_0 + I$, which is not allowed by the syntax of \qtwomlo{}.  
In~\cite{Rabinovich2010}, it is claimed (without proof) that in the continuous semantics,
\mitl{} extended with such modalities (\tlci{}) is equally expressive as the fragment with only the most basic versions of the counting modalities
(allowing only $I = (0, 1)$).
By contrast, Krishna \emph{et al.}~\cite{KriMad16} showed that in the pointwise semantics, $\counting^k_I$ with $I = \langle a, b \rangle$ cannot be expressed in the \emph{future} fragment of \tlci{}
with only counting modalities with $I = [0, b \rangle$.
In this paper, we reconcile these results and reaffirm the claim, i.e.~we prove that $\counting^k_I$ with $I = \langle a, b \rangle$ is indeed expressible in (future) \qtwomlo{} in both the pointwise and continuous semantics.
This suggests that \qtwomlo{} is a very expressive and robust logic in both the pointwise and continuous semantics. 
From~\cite{HoMadnani23}, we also know that in the pointwise semantics,
$\counting^k_I$ with $I = \langle a, b \rangle$ is expressible in the fragment of \tlci{} with  (both future and past) counting modalities with $I = [0, b \rangle$. 

\paragraph{Contributions.}
We argue that the folklore
belief---$\counting_I^{k}$ with $I = \langle a, b \rangle$
can be rewritten into formulae using only $\counting_I^{k}$ with $I = [0, b \rangle$ in about the same way as $\until_I$ with $I = \langle a, b \rangle$
can be rewritten into $\until_I$ with $I = [0, b \rangle$---is not correct.
We however show that by allowing \emph{automata modalities} (or, equivalently, \qtwomlo{} or \qtwomso{}~\cite{KrishnaMP18}), 
one can indeed enforce that a sequence of events specified lies in the required interval;
the proof is based on a generalisation of the 
techniques developed in~\cite{Ho19} to show that \emph{Extended Metric Interval Temporal Logic}
(\emitl{}~\cite{Wilke1994}) remains as expressive when restricted to only \emph{unilateral} intervals, i.e.~in the form of~$[0, b \rangle$ or $\langle a, \infty)$.
Building upon this insight, we `correct' the folklore belief by showing that $\counting_I^{k}$ with $I = \langle a, b \rangle$ can actually be expressed in $\counting_I^{k}$ with $I = [0, b \rangle$ (without using $\pcounting_I^k$) in a more involved way (in the pointwise semantics as well, under some extra conditions).

\paragraph{Related work.}
Hirshfeld and Rabinovich~\cite{HirRab99, HirshfeldR99, Hirshfeld2004, Hirshfeld06, hirshfeld2008decidable, Rabinovich2010, HirshfeldR12} pioneered the research on decidable timed logics that extends \mitl{} with counting and Pnueli modalities, which culminates in the strong metric predicate logic \qtwomlo{}.  
Hunter~\cite{Hunter2013} later proved that if \mtl{}~\cite{Koy90} (which is exactly like \mitl{}, but singular $I$'s are allowed) is extended in the same way, or equivalently if singular $I$'s are allowed in \qtwomlo{}, one obtains a logic that is expressively complete for \foone{} (in the continuous semantics).

In the context of temporal logics and model checking, there are also some closely related results that are not directly comparable with the present paper. Extending \ltl{} with \emph{threshold counting} is first done by Laroussinie \emph{et al.}~\cite{Laroussinie2010} where the `until' ($\until$) modality is extended with counting specifications. The timed versions of such modalities $\mathbf{UT}_I$ are studied by Krishna \emph{et al.} in~\cite{KriMad16}. 
Another type of counting specification is \emph{modulo counting}, which counts the number of events (seen so far) satisfying some monadic predicate modulo a given constant $N$. \ltl{} extended with modulo counting modalities is first considered by Baziramwabo \emph{et al.}~\cite{782629}, and Lodaya and Sreejith~\cite{LodayaS10} showed that $N$ can be encoded succinctly yet still retaining the $\mathrm{PSPACE}$ upper bound.
Bednarczyk and Charatonik~\cite{conf/fsttcs/BednarczykC17} studied the complexity of the satisfiability problem of the two variable fragment of first-order logic extended with modulo counting quantifiers interpreted over both trees and words.
Similar operations also appear in other contexts, such as \emph{temporal aggregation}~\cite{Bellomarini2021MonotonicAF} in databases and knowledge graphs.


\section{Preliminaries}\label{sec:prelim}
We give a brief account of the required background on timed logics. For more detailed reviews and comparisons of relevant results, we refer the readers to~\cite{Hirshfeld2004, BouyerLMOW17}.
Note that, in contrast with~\cite{Wilke1994, Hirshfeld06, Rabinovich2010}, we focus mainly on the \emph{future} fragments of metric temporal logics.

\paragraph{Timed languages.}
A \emph{timed word} over a finite alphabet $\Sigma$ is an $\omega$-sequence
of \emph{events} $(\sigma_i,\tau_i)_{i \geq 1}$ over
$\Sigma \times \R_{\geq 0}$ with $(\tau_i)_{i\geq 1}$ a non-decreasing
sequence of non-negative real numbers (`\emph{timestamps}') such that for each $r \in \R_{\geq 0}$,
there is some $j \geq 1$ with $\tau_j \geq r$ (i.e.~we require all timed words
to be `\emph{non-Zeno}').
We denote by $T\Sigma^\omega$ the set of all the timed words over $\Sigma$. A \emph{timed language} is a
subset of $T\Sigma^\omega$.

\paragraph{Metric predicate logics.}
\emph{Monadic Second-Order Logic of Order and Metric} (\msoone{})~\cite{Alur1993, Wilke1994}
formulae over a finite set of atomic propositions (monadic predicates) $\AP$
are generated by 
\begin{displaymath}
\vartheta  ::=  \top \,\mid\, X(x) \,\mid\, x < x' \,\mid\, d(x, x') \in I \,\mid\, \vartheta_1 \wedge \vartheta_2 \,\mid\, \neg \vartheta \,\mid\,
\exists x \, \vartheta
\,\mid\,
\exists X \, \vartheta
\end{displaymath}
where $X \in \AP$, $x, x'$ are first-order variables, $d$ is the distance predicate, 
$I \subseteq \R_{\geq 0}$ is an interval with endpoints in $\N \cup\{\infty\}$, 
and $\exists x$, $\exists X$ are first- and second-order quantifiers, respectively.
We write, e.g., $(a, b \rangle$, 
to refer to $(a, b)$ or $(a, b]$.
We say that $x$ (respectively $X$) is a \emph{free} first-order (respectively second-order) variable in $\vartheta$ if it does not appear in the scope of $\exists x$ (respectively $\exists X$) in $\vartheta$.
We usually write $\vartheta(x_1, \dots, x_{m}, X_1, \dots, X_n)$ for $\vartheta$, if $x_1$, \dots, $x_m$
and $X_1$, \dots, $X_n$ are free in $\vartheta$.
We say that an \msoone{} formula $\vartheta(x)$ with only a free first-order variable $x$ is a \emph{future formula} if 
all the quantifiers appearing in $\vartheta(x)$ are relativised to $(x, \infty)$, i.e.~if $\exists x' \, \theta$ (respectively $\forall x' \, \theta$)
is a subformula of $\vartheta(x)$, then $\theta$ is of the form $x < x' \wedge \theta'$ (respectively $x < x' \implies \theta'$).
The fragment of \msoone{} without second-order quantifiers is the  \emph{Monadic First-Order Logic of Order and Metric} (\foone{}).
The fragment of \foone{} without the distance predicate is the
\emph{Monadic First-Order Logic of Order} (\fo{}).
\qtwomlo{}~\cite{HirRab99} is a fragment of 
\foone{} obtained
from \fo{} by allowing only non-singular $I$'s (for the sake of decidability~\cite{Alur1993, Ouaknine2006}) and a restricted use of distance predicates. More precisely, \qtwomlo{} is the smallest syntactic fragment of \foone{} satisfying the following conditions:
\begin{itemize}
\item All \fo{} formulae $\vartheta(x)$ with only a free first-order variable $x$ are \qtwomlo{} formulae.
\item If $\vartheta(x_0, x)$ is an \fo{} formula (possibly with \qtwomlo{} formulae used as monadic predicates) where $x_0$, $x$ are the only free first-order variables, then 
\begin{itemize}
\item $\exists x \, \big(x_0 < x 
\wedge d(x_0, x) \in I \wedge \vartheta(x_0, x)\big)$ and 
\item $\exists x \, \big(x < x_0 
\wedge d(x_0, x) \in I \wedge \vartheta(x_0, x)\big)$,
\end{itemize}
where $I$ is non-singular,
are also \qtwomlo{} formulae (with free first-order variable $x_0$).
\end{itemize}
The future fragment $\qtwomlo{}^\textsf{fut}$ is obtained by allowing only $\vartheta(x)$ and
$\exists x \, \big(x_0 < x 
\wedge d(x_0, x) \in I \wedge \vartheta(x_0, x)\big)$ above and 
also requiring them to be future formulae.
In the same way we can define the corresponding fragments \mso{}, \qtwomso{}, and $\qtwomso{}^\textsf{fut}$~\cite{KrishnaMP18} of \msoone{}.

\paragraph{Metric temporal logics.}

A \emph{non-deterministic finite automaton} (\nfa{}) over $\Sigma$
is a tuple $\A = \langle \Sigma, S, s_0,\transitions, F \rangle$
where $S$ is a finite set of locations, $s_0 \in S$ is the initial location,
$\transitions \subseteq S \times \Sigma \times S$ is the transition relation,
and $F$ is the set of final locations.
We say that $\A$ is \emph{deterministic} (a \dfa{}) iff for each $s \in S$
and $\sigma \in \Sigma$, $| \{ (s, \sigma, s') \mid (s, \sigma, s') \in \transitions \} | \leq 1$. 
A \emph{run} of $\A$ on $\sigma_1 \dots \sigma_n \in \Sigma^+$
 is a
sequence of locations $s_0 s_1 \dots s_n$ where 
there is a transition $(s_i,\sigma_{i+1},s_{i+1}) \in \transitions$
for each $i$, $0 \leq i < n$.  A run of $\A$ is \emph{accepting} iff
it ends in a final location. A finite word is \emph{accepted} by $\A$
iff $\A$ has an accepting run on it.

(Future) \emph{Extended Metric Interval Temporal Logic} ($\emitl{}^\textsf{fut}$)~\cite{Wilke1994} formulae over
a finite set of atomic propositions $\AP$
are generated by  
\begin{displaymath}
  \phi ::= \top \,\mid\, P \,\mid\, \phi_1 \land \phi_2 \,\mid\, \neg\phi \,\mid\, \A_I(\phi_1, \dots, \phi_n)
\end{displaymath}
where $P \in \AP$, $\A$ is an \nfa{} over the $n$-ary alphabet $\{ 1, \dots, n \}$,
and $I \subseteq \R_{\geq 0}$ is a non-singular interval with endpoints in $\N \cup\{\infty\}$.\footnote{For notational simplicity,
we also use $\phi_1$,~\dots, $\phi_n$ directly as transition labels (instead of $1$, \dots, $n$) in the figures.}
We sometimes omit the subscript $I$ when $I = [0, \infty)$ and write pseudo-arithmetic expressions for lower
or upper bounds, e.g.,~`$< 3$' for $[0, 3)$.
We also omit the arguments $\phi_1$,~\dots, $\phi_n$
and simply write $\A_I$, if clear from the context.
(Future) \emph{Metric Interval Temporal Logic} ($\mitl{}^\textsf{fut}$)~\cite{AluFed96} is the fragment of $\emitl{}^\textsf{fut}$
with only the `until' modalities
defined by the \nfa{} $\A^{\until}$ in~\figurename~\ref{fig:until.nfa}
(usually written in infix notation as $\phi_1 \until_I \phi_2$).
We also use the usual shortcuts like
$\bot \equiv \neg \top$,
$\nextx_I \varphi \equiv \bot \until_I \varphi$,
$\eventually_I\phi \equiv \top\until_I\phi$, 
$\weventually_I\phi \equiv \phi \vee \eventually_I \phi$, 
$\globally_I\phi \equiv \neg\eventually_I\neg\phi$,
and $\phi_1\release_I \phi_2 \equiv \neg\big((\neg\phi_1)\until_I(\neg \phi_2)\big)$.
(Future) \emph{Linear Temporal Logic} ($\ltl{}^\textsf{fut}$)~\cite{Pnueli1977} is the
fragment of $\mitl{}^\textsf{fut}$ where all modalities are labelled by $[0,\infty)$.
$\tlc{}^\textsf{fut}$~\cite{Hirshfeld06} is the fragment of $\emitl{}^\textsf{fut}$
obtained from $\mitl{}^\textsf{fut}$ by adding the 
\emph{counting modalities}
$\counting_I^{k}$, where $I$
is a non-singular upper-bound interval (i.e.~of the form $[0, b \rangle$ for some $b \in \N_{> 0} \cup\{\infty\}$)
and $k \geq 1$.\footnote{This definition is a mild generalisation of the modalities $\counting_I$ in~\cite{Hirshfeld06, hirshfeld2008decidable} where $I$ must be $(0, 1)$. Note that
\tlc{} is equivalent to the unilateral fragment of \tlci{} (defined later in Section~\ref{sec:counting}), as intervals
of the form $\langle a, \infty)$ can easily be eliminated in general.}
For example, $\counting_I^{3} \varphi$ (`$\varphi$ happens at least $3$ times in $I$ in the future') is 
defined by the \nfa{} $\A^{\counting, 3}$
 in~\figurename~\ref{fig:counting.nfa}.

The definitions above are for the future versions of the modalities,
but we note that we can also define the \emph{past} versions of the modalities
and correspondingly the full fragments of logics (denoted by names with no `\textsf{fut}' superscripts), e.g.,~\emitl{}~\cite{Wilke1994} and \mitl{}~\cite{AlurH92}.

\begin{figure}[t]
\centering
\begin{minipage}{0.45\linewidth}
\centering
	\scalebox{.75}{\begin{tikzpicture}[node distance = 2.5cm]
		\node[initial left, state] (0) {};
		\node[state, right of=0] (1) {};
		\node[state, accepting, right of=1] (2) {};
		\path
		(0) edge[->] (1)
		(1) edge[loopabove, ->] node[above=1mm]{$\varphi_1$} (1)
		(1) edge[->] node[above=1mm]{$\varphi_2$} (2)
		(2) edge[loopabove, ->] (2);
	\end{tikzpicture}}
\caption{The \nfa{} $\A^{\until}$ for $\varphi_1 \until_I \varphi_2$.}
\label{fig:until.nfa}
\end{minipage}
\begin{minipage}{0.45\linewidth}
\centering
	\scalebox{.75}{\begin{tikzpicture}[node distance = 2cm]
		\node[initial left ,state] (0) {};
		\node[state, right of=0] (1) {};
		\node[state, right of=1] (2) {};
		\node[state, right of=2] (3) {};
		\node[state, accepting, right of=3] (4) {};
		\path
		(0) edge[->] (1)
		(1) edge[loopabove, ->] (1)
		(1) edge[->] node[above=1mm]{$\varphi$} (2)
		(2) edge[loopabove, ->] (2)
		(2) edge[->] node[above=1mm]{$\varphi$} (3)
		(3) edge[loopabove, ->] (3)
		(3) edge[->] node[above=1mm]{$\varphi$} (4)
		(4) edge[loopabove, ->] (4);
	\end{tikzpicture}}
\caption{The \nfa{} $\A^{\counting, 3}$ for $\counting_I^{3} \varphi$.}
\label{fig:counting.nfa}
\end{minipage}
\end{figure}

\paragraph{Semantics.}
With each timed word $\rho = (\sigma_i,\tau_i)_{i \geq 1}$ over $\Sigma_{\AP}= 2^\AP$ we associate a structure $M_\rho$
whose universe $U_\rho$ is $\{ i \mid i \geq 1\}$.
The order relation $<$ and atomic propositions in $\AP$ are interpreted in the expected way, e.g.,~$P(i)$ holds in $M_\rho$ iff $P \in \sigma_i$.
The distance predicate $d(x, x') \in I$ holds iff $|\tau_x - \tau_{x'}| \in I$.
The satisfaction relation for 
\msoone{} is defined inductively as usual:
we write $M_\rho, j_1, \dots, j_{m}, J_1, \dots, J_{n} \models \vartheta(x_1, \dots, x_{m}, X_1, \dots, X_n)$
(or simply $\rho, j_1, \dots, j_{m}, J_1, \dots, J_n \models \vartheta(x_1, \dots, x_{m}, X_1, \dots, X_n)$)
if $j_1, \dots, j_{m} \in U_\rho$,
$J_1, \dots, J_{n} \subseteq U_\rho$,
and $\vartheta(j_1, \dots, j_{m}, J_1, \dots, J_n)$ holds in $M_\rho$.
We say that two \msoone{} formulae $\vartheta_1(x)$ and $\vartheta_2(x)$ are \emph{equivalent} if for all timed words $\rho = (\sigma_i,\tau_i)_{i \geq 1}$ and $j \in U_\rho$,
\[
\rho, j \models \vartheta_1(x) \iff \rho, j \models \vartheta_2(x) \;.
\]

Given a $\emitl{}^\textsf{fut}$ formula $\phi$ over $\AP$, 
a timed word $\rho=(\sigma_i,\tau_i)_{i \geq 1}$ over $\Sigma_\AP = 2^\AP$ and
a \emph{position} $i \geq 1$,
we define the satisfaction relation $\rho, i \models \varphi$ as follows:
\begin{itemize}
\item $\rho, i \models \top$;
\item $\rho, i \models p$ iff $p\in \sigma_i$;
\item $\rho, i \models \phi_1\land \phi_2$ iff $\rho,i \models\phi_1$ and $\rho,i \models\phi_2$;
\item $\rho,i \models\neg\phi$ iff $\rho,i \not\models\phi$;
\item $\rho,i \models \A_I(\phi_1, \dots, \phi_n)$ iff there exists
  $j\geq i$ such that (i) $\tau_j-\tau_i\in I$ and (ii) there is an
accepting run of $\A$ on $a_i \dots a_j$ where 
$\rho, \ell \models \phi_{a_\ell}$ ($a_\ell \in \{1, \dots, n\}$) for each $\ell$, $i \leq \ell \leq j$.
\end{itemize}
We say that $\rho$ \emph{satisfies} $\phi$ (written $\rho\models\phi$)
iff $\rho,1 \models\phi$.

The definitions above correspond to the so-called \emph{pointwise} semantics of timed logics~\cite{Alur1993, AluHen94, Wilke1994, OuaWor07}.
It is also possible to define the \emph{continuous} semantics of timed logics over timed words by taking $\R_{\geq 0}$ instead of $\{ i \mid i \geq 1\}$ as the universe and $d(x, x') = |x - x'|$; we refer the readers to~\cite{DSouza2007,Bouyer2010,OuaknineRW09} for details. 
While we focus on the former in this paper, it is clear that all of our results carry over to the continuous interpretations of timed logics
where system behaviours are modelled as (finitely variable) signals.

\paragraph{Expressiveness.}
We say that a metric logic $L'$ is \emph{expressively complete} for a metric logic $L$ iff for any formula $\vartheta(x) \in L$,
there is an equivalent formula $\varphi(x) \in L'$.\footnote{Formulae of metric temporal logics in this paper are
\msoone{} formulae with a single free first-order variable.
}
We say that $L'$ is \emph{at least as expressive as}
(or \emph{more expressive than}) $L$ (written $L \subseteq L'$)
iff for any formula $\vartheta(x) \in L$, there is an \emph{initially equivalent} formula $\varphi(x) \in L'$
(i.e.,~$\vartheta(1)$ and $\varphi(1)$ evaluate to the same truth value for any timed word).
If $L \subseteq L'$ but $L' \nsubseteq L$ then we say that $L'$ is \emph{strictly more expressive than} $L$
(or $L$ is \emph{strictly less expressive than} $L'$).
We write $L \equiv L'$ iff $L \subseteq L'$ and $L' \subseteq L$. For the purpose of this paper, the most relevant
known expressiveness results are  $\emitl{}^\textsf{fut} \equiv \qtwomso{}^\textsf{fut}$
and $\text{\emph{aperiodic}~\cite{SchutzenbergerSyntacticMonoid, McNaughtonStarFreeLanguages} } \emitl{}^\textsf{fut} \equiv \qtwomlo{}^\textsf{fut}$~\cite{KrishnaMP18}, and thus we will freely mix the use of them.

\section{Expressing counting modalities}\label{sec:counting}

\paragraph{Counting events in arbitrary intervals.}
We start by giving an alternative and more general definition (in terms of \foone{}) of what do we mean by
counting events in an interval $I$.
Note that the following definition 
of $\counting_I^{k} \varphi$ is equivalent to the definition based on automata modalities in Section~\ref{sec:prelim}
for the special case where $I$ is of the form $[0, b\rangle$.
\begin{definition}[$\tlci{}^\textsf{fut}$~\cite{Rabinovich2010}]
$\tlci{}^\textsf{fut}$ is obtained from $\mitl{}^\textsf{fut}$
by adding the (one-place) modalities
$\counting_I^{k}$ defined by the following formula (where $I$ is non-singular): 
\[
    \vartheta^{\counting, k}_{I}(x, X) = \exists x_1 \, \dots \, \exists x_k \, \big(x < x_1 < \dots < x_k \wedge d(x, x_1) \in I \wedge d(x, x_k) \in I \wedge \bigwedge_{1 \leq i \leq k} X(x_i) \big) \;.
\]
\tlci{} is obtained by adding the past counterparts of the modalities above (defined symmetrically).
\end{definition}
We first note that while $\vartheta^{\counting, k}_{I}(x, X)$ is in \foone{}, it is not in \qtwomlo{} (at least syntactically), thus it is not immediately clear how to 
express it in \tlc{} (with both the future and past modalities) even in the continuous semantics, as the translation from \qtwomlo{} to \tlc{} in~\cite{Hirshfeld06, HirshfeldR12} does not apply.
It should also be clear that the trivial attempt of simply decorating $\A^{\counting, k}$ with an arbitrary non-singular $I$ would not give a formula equivalent to 
$\vartheta^{\counting, k}_{I}$. For example, the following timed word
\[
(\emptyset, 0)(\{P\}, 0.5)(\{P\}, 1.5)(\{P\}, 2.5)(\{P\},3.5)\dots
\]
satisfies $\A^{\counting, 3}_{(2, 3)} P$, but clearly 
$\rho, 1 \centernot \models \vartheta^{\counting, 3}_{(2,3)}(x, X)$.
In~\cite{Rabinovich2010}, it is stated that \tlc{}
is as expressive as \tlci{}, but no complete proof is given. In~\cite{FR08-TR2008-10} the following equivalence, which is reminiscent of how \mitl{} and \qtwomlo{} with arbitrary non-singular intervals can be reduced to their base versions using only $I = (0, 1)$ in the continuous semantics~\cite{HenRas98, HirshfeldR99, Hirshfeld06}, is proposed:
\begin{equation}\label{eqn.wrong}
\textstyle{
\counting_{(a, a + 1)}^{k} P \iff \globally_{(0, 1)} \eventually_{(0, a)}    \counting_{(0, 1)}^{k} P
} \;.
\end{equation}
This is, however, not correct in either the pointwise or the continuous semantics---for instance, if $k = 2$
and $a = 2$, then any timed word with only one event at $\tau_1 + 1$, two $P$-events in $\tau_1 + (1, 2)$, and no $P$-event in $\tau_1 + (2, 3)$ satisfies the right-hand side of (\ref{eqn.wrong}), but not its left-hand side; if $k = 2$ and $a = 1$, then
\[
(\emptyset, 0)(\emptyset, 0.6)(\emptyset, 0.7)(\{P\}, 0.8)(\{P\}, 0.9)(\emptyset, 1.6)(\{P\},1.7)(\{P\},2.1) \dots
\]
satisfies the right-hand side of (\ref{eqn.wrong}), but not its left-hand side.

In the study of timed logics, it is common to rule out constraints involving singular 
(`\emph{punctual}') intervals as they can easily render the \emph{satisfiability problem} undecidable (or have prohibitively high complexity~\cite{OuaWor07}).
If we do however allow singular intervals, then 
the following equivalence clearly holds in the continuous semantics:
\begin{equation}\label{eqn.punctual}
\textstyle{
\counting_{(a, a + 1)}^{k} P \iff \eventually_{= a}  \counting_{(0, 1)}^{k} P
} \;.
\end{equation}
Indeed, the main difficulty in expressing (\ref{eqn.punctual})
in \tlc{} is the lack of ability to express punctuality---roughly speaking, $\globally_{(0, 1)} \eventually_{(0, 1)} \varphi$
is a weaker requirement than $\eventually_{= 1} \varphi$:
the former is also satisfied by two points that both satisfy $\varphi$, 
surround $t + 1$ (where $t$ is the current time), and separated by less than $1$.
Therefore, while $\globally_{(0, 1)} \eventually_{(0, 1)} \counting_{(0 ,1)}^k P$ implies $\eventually_{(1, 2)} \counting_{(0 ,1)}^k P$ 
or $\eventually_{=1} \counting_{(0 ,1)}^k P$,
it does not guarantee that all the $k$ `witnesses' lie within $t + (1, 2)$ in the former case.
On the other hand, 
$\eventually_{(0, 1)} \globally_{(0, 1)} \counting_{(0, 1)}^k P$
does not necessarily hold when $\eventually_{=1} \counting_{(0, 1)}^k P$ holds,
as $\eventually_{(0, 1)} \globally_{(0, 1)} \varphi$
is a stronger requirement than $\eventually_{=1} \varphi$.

Before we explain how to express 
$\counting_I^{k}$ for the general case where $I = \langle a, b\rangle$ with $a < b$ in $\qtwomlo{}^\textsf{fut}$
in the next section, 
let us first mention two simple ways that do not involve punctuality to express  
them in
non-trivial extensions of \mitl{}.
\paragraph{Counting events in $I$ by automata modalities.}
In the case of counting where each witness is `context free', instead of trying to locate a suitable point where $\counting_{(0, 1)}^k P$ holds (like in (\ref{eqn.wrong})), we can specify that there are $k$ \emph{distinct} $P$-events in $t + (a, a+1)$---this can be done with $k$ modulo-$k$ counters,
similar to an idea used in~\cite{KriMad16}.
For example, if $k = 3$ we use three automata modalities that accept every $(3n)$-th, $(3n+1)$-th, and $(3n+2)$-th $P$-event, respectively, and then specify that each of them has a run that ends in $t + (a, a+1)$.
The following theorem is then immediate.
\begin{theorem}
$\tlci{}^\textsf{\textup{fut}} \subseteq \text{aperiodic } \emitl{}^\textsf{\textup{fut}} \equiv \qtwomlo{}^\textsf{\textup{fut}}$.
\end{theorem}
\noindent This idea, however, does not easily generalise to \tlpi{}, which we discuss in the next section.

\paragraph{Counting events in $I$ by rational constants.}

Recall from~\cite{HunterOW13} that  
$\counting_{(0, 1)}^2 P$ can be expressed as the disjunction of 
     $\eventually_{(0, \frac{1}{2})} (P \wedge \eventually_{(0, \frac{1}{2})} P)$,
    $\eventually_{(\frac{1}{2}, 1)} (P \wedge \once_{(0, \frac{1}{2})} P)$, and
    $\eventually_{(0, \frac{1}{2})} P \wedge \eventually_{(\frac{1}{2}, 1)} P$ (where $\once$ is the past version of $\eventually$).
This can easily be generalised (like in~\cite{HunterOW13}, but with trivial modifications to avoid using punctualities) to arbitrary non-singular $I$ and larger values of $k$, e.g., for $\counting_{(1, 2)}^3 P$, we partition $(1, 2)$ into $6$ subintervals and consider the cases where \begin{enumerate*}[label=\arabic*)] \item all three witnesses lie within one of the three subintervals covering $(1, 1.5)$; \item all three witnesses lie within one of the three subintervals covering $(1.5, 2)$; and \item not all witnesses lie within a single subinterval. \end{enumerate*}

\begin{theorem}\label{thm:rational}
\mitl{} (with both the future and past modalities) is expressively complete for \tlci{}, if rational constants are allowed.   
\end{theorem}

This also applies straightforwardly to \tlpi{}.
On the other hand, \mitl{} with \emph{only one of} these extensions---i.e.~either past modalities~\cite{PandyaS11} or rational constants~\cite{Bouyer2010}---
is insufficient for expressing \tlpi{}.

\section{Expressing $\pnueli^2_I$ in $\qtwomlo{}^\textsf{\textmd{fut}}$}\label{sec:pnueli}

A more general form of counting, where one can specify a sequence of distinct events, is enabled by the \emph{Pnueli modalities} $\pnueli_I^k$ defined below. 
Once again,~\cite{Rabinovich2010} states that they are expressible in \tlc{} without proof.
\begin{definition}[$\tlpi{}^\textsf{fut}$~\cite{Rabinovich2010}]
$\tlpi{}^\textsf{fut}$ is obtained from $\mitl{}^\textsf{fut}$
by adding the ($k$-place) modalities
$\pnueli_I^k$ defined by the following formula (where $I$ is non-singular): 
\[
    \vartheta^{\pnueli, k}_{I}(x, X_1, \dots, X_k) = \exists x_1 \, \dots \, \exists x_k \, \big(x < x_1 < \dots < x_k \wedge d(x, x_1) \in I \wedge d(x, x_k) \in I \wedge \bigwedge_{1 \leq i \leq k} X_i(x_i) \big) \;.
\]
\tlpi{} is obtained by adding the past counterparts of the modalities above (defined symmetrically).
\end{definition}

The modulo-$k$ trick that we used earlier to express $\counting_I^k$ no longer works in the case of Pnueli modalities, as obviously we must also
ensure that $X_1, \dots, X_k$ are satisfied \emph{in this order} by a sequence of events in $I$.
We now describe a general construction of $\qtwomlo{}^\textsf{fut}$ formulae (or, equivalently, aperiodic $\emitl{}^\textsf{fut}$ formulae where all automata modalities are definable by $\ltl{}^\textsf{fut}$ or future \fo{} formulae~\cite{KrishnaMP18}) that specify sequences of events in arbitrary non-singular intervals. 
For simplicity, we will use $\pnueli^2_{(a, a+1)}(P, Q)$ with $a \geq 1$ as an example to explain the ideas involved 
before we extend the construction to the general case where the sequence of events is specified by a first- or second-order formula  in the next section.


Let us call a pair of positive integers $\langle h, \ell \rangle$ where $h \leq \ell$ a \emph{segment}.
Given a timed word 
$\rho = (\sigma_i,\tau_i)_{i \geq 1}$ over $\Sigma_{\AP}$ where $\AP = \{ P, Q \}$, we say that a segment $\langle h, \ell \rangle$ is a \emph{witness} for $\pnueli^2_{(a, a+1)}(P, Q)$ at $i$ if $h < \ell$, $P \in \sigma_h$, $Q \in \sigma_\ell$, 
 $\langle h, \ell \rangle$ is \emph{minimal} in the sense that there is no $h', \ell'$ such that $h \leq h' \leq \ell' \leq \ell$, either $h < h'$ or $\ell' < \ell$, and $\langle h', \ell \rangle$  also satisfies the conditions above, and both $\tau_h, \tau_\ell \in \tau_i + (a, a+1)$.
In other words, $\rho, h \models \exists x' \, \varphi_1(x, x')$ where
\[
\varphi_1(x, x') = x < x' \wedge P(x) \wedge Q(x') \wedge \neg \exists y \, \Big(x < y < x' \wedge \big(P(y) \vee Q(y)\big)\Big) \;.
\]
The idea is that
$\exists x' \, \varphi_1(x, x')$
holds at the starting points $h$ of all the \emph{potential witnesses} (witnesses but without the timing requirement in relation to $\tau_i$) for $\pnueli^2_{(a, a+1)}(P, Q)$.
For each $i \geq 1$, we either have $\rho, i \models \exists x' \, \varphi_1(x, x')$ or $\rho, i \centernot \models \exists x' \, \varphi_1(x, x')$, and 
this gives rise to a (finite or infinite) sequence of 
potential witnesses for $\pnueli^2_{(a, a+1)}(P, Q)$:
\[
\langle h_1, \ell_1 \rangle \langle h_2, \ell_2 \rangle\dots
\]
where $h_1 < h_2 < \dots$.
From the definition of $\varphi_1$, it is clear that $\ell_j \leq h_{j+1}$ for all $j$ (i.e.~the potential witnesses for $\pnueli^2_{(a, a+1)}(P, Q)$ do not overlap except possibly on the endpoints).

Now, to specify that $\rho, i \models \pnueli^2_{(a, a+1)}(P, Q)$, we want to express the condition that some potential witness $\langle h_j, \ell_j \rangle$ for $\pnueli^2_{(a, a+1)}(P, Q)$
actually satisfies the timing requirement $\tau_{h_j}, \tau_{\ell_j} \in \tau_i + (a, a+1)$.
We start from this initial attempt to express 
$\pnueli^2_{(a, a+1)}(P, Q)$:
\[
\varphi_{\textit{wit}} = \eventually_{(a, a+1)} \varphi_1 \wedge
\A^1_{(a, a+1)} 
\]
where $\varphi_1$ is the \ltl{} formula equivalent to $\exists x' \, \varphi_1(x, x')$, $\A^1$ is the equivalent \nfa{} for
$\varphi_1'(x, x') = \exists y \,  \big( x < y < x' \wedge \varphi_1(y, x') \big) $.\footnote{Technically, we can use a theorem in~\cite{Gabbay1980} to get equivalent finite-word \ltl{} formulae (over infinite-word \ltl{} formulae as monadic predicates) for
\fo{} formulae of the form $\varphi(x, x')$.
}
Intuitively, $\eventually_{(a, a+1)} \varphi_1$ says that $d(i, h_j) \in (a, a+1)$ for some $j$, and 
$\A_{(a, a+1)}$ says that
$d(i, \ell_{j'}) \in (a, a+1)$ for some $j'$.
But it is not hard to see that an undesired scenario (illustrated in~\figurename~\ref{fig:scenario.one}), where no potential witness $\langle h, \ell \rangle$ for $\pnueli^2_{(a, a+1)}(P, Q)$ lies completely within $\tau_i + (a, a+1)$,
also satisfies $\varphi_{\textit{wit}}$.
To capture and rule out this undesired scenario, note that in~\figurename~\ref{fig:scenario.one} it is clear that 
the time elapsed between $h_j$ and $\ell_{j+1}$ is greater or equal than $1$. Based on this observation, we can write a formula involving the two adjacent potential witnesses $\langle h_j, \ell_j \rangle$ and 
$\langle h_{j+1}, \ell_{j+1} \rangle$ for $\pnueli^2_{(a, a+1)}(P, Q)$:
\[
\varphi_2(x, y') = \exists x' \, \exists y \, \Big(
x < y \wedge x \leq x' \wedge y \leq y' \wedge
\varphi_1(x, x') \wedge \varphi_1(y, y') \wedge \neg \exists z \, \exists z' \, \big(x < z < y \wedge z \leq z' \wedge
 \varphi_1(z, z')\big) \Big) \;.
\]
To express $d(h_j, \ell_{j+1}) \geq 1$, we just check if the $\qtwomlo{}^\textsf{fut}$ formula
\[
\varphi_2^{\geq 1}(x) = \exists y' \, \big(x < y' \wedge d(x, y') \geq 1 \wedge \varphi_2(x, y')\big)
\]
holds at position $h_j$. It remains to enforce the following conditions:
\begin{itemize}
    \item $\langle h_j, \ell_j \rangle$ is the \emph{last} segment $\langle h, \ell \rangle$ with $\tau_{h} \leq \tau_i + a$.
    \item $\tau_{\ell_{j+1}} \geq \tau_i + (a+1)$; see~\figurename~\ref{fig:scenario.one.desired} for an example
    when $\langle h_{j+1}, \ell_{j+1} \rangle$ lies completely within $\tau_i + (a, a+1)$ but $\varphi_2^{\geq 1}(x)$ holds at $h_j$.  
\end{itemize}
We now use the following crucial lemma to locate the last 
$\langle h, \ell \rangle$ with $\tau_{h} \leq \tau_i + a$.
\begin{lemma}
For any $\rho = (\sigma_i,\tau_i)_{i \geq 1}$ over $\Sigma_{\AP}$ where $\AP = \{ P, Q \}$,
the $\qtwomlo{}^\textsf{\textup{fut}}$
formula $\varphi_2^{\geq 1}(x)$ is satisfied by at most $2a + 2$ positions $j > i$ with $d(i, j) \in [0, a]$ for any $i \geq 0$.
\end{lemma}
\begin{proof}
Let $\langle h_1, \ell_1 \rangle \langle h_2, \ell_2 \rangle\dots$ be the sequence of 
potential witnesses for $\pnueli^2_{(a, a+1)}(P, Q)$ as described above.
If $\rho, h_j \models \varphi_2^{\geq 1}(x)$,
then either there is no $\langle h_{j+2}, \ell_{j+2} \rangle$ or $\tau_{h_{j+2}} \geq \tau_{h_{j}} + 1$.
It follows that if there are $2a + 3$ positions satisfying $\varphi_2^{\geq 1}(x)$, then the first and the last of them must be more than $a$ apart.
\end{proof}
\noindent It follows that the undesired scenario \#1 is captured by
\[
\varphi_{\text{out}} =
\bigvee_{1 \leq k \leq 2a + 2} 
\big(
\textstyle{\counting_{\leq a}^k} (\A^2_{\geq 1}) \wedge \neg  \textstyle{\counting_{\leq a}^{k + 1}} (\A^2_{\geq 1}) 
\wedge
\B^k_{\geq a + 1}
\big)
\]
where 
$\A^2$ is the equivalent \nfa{} for
$\varphi_2(x, y')$ (i.e.~$\A^2_{\geq 1} \equiv \varphi_2^{\geq 1}(x)$)
and $\B^k$ is the equivalent \nfa{} for
\begin{IEEEeqnarray*}{rCll}
\varphi_2^k(x, x') & = 
 & \exists x_1 \, \dots \exists x_k \, 
\Big( & x < x_1 < \dots < x_k < x' \wedge \varphi_2^{\geq 1}(x_1) \wedge \dots \wedge \varphi_2^{\geq 1}(x_k) \wedge \varphi_2(x_k, x') \\
& & & {} \wedge \neg \exists y \, \big(x \leq y \leq x_k \wedge \bigwedge_{1 \leq j \leq k} (y \neq x_k) \wedge \varphi_2^{\geq 1}(y)\big) \Big) \;;
\end{IEEEeqnarray*}
it can be obtained by regarding $\varphi_2^{\geq 1}$ as an atomic proposition and replace it afterwards by $\A^2_{\geq 1}$. 
Specifically, the first two conjuncts specify that the number of positions satisfying $\varphi_2^{\geq 1}(x)$ before $\tau_i + a$ is exactly $k$, and the last conjunct ensures that the second potential witness 
in this pair is out of bounds, i.e.~$\tau_{\ell_{j+1}} \geq \tau_i + (a + 1)$.
The desired formula is 
\[
\varphi^{\pnueli, 2}_{(a,a+1)}(P, Q) =
\varphi_{\textit{wit}} \wedge \neg \varphi_{\textit{out}} \;.
\]
\begin{proposition}
$\varphi^{\pnueli, 2}_{(a,a+1)}(P, Q) \equiv \pnueli^2_{(a, a+1)}(P, Q)$.
\end{proposition}
\begin{proof}
If $\varphi^{\pnueli, 2}_{(a,a+1)}(P, Q)$ holds at $i$ then either there is a potential witness $\langle h, \ell \rangle$ for $\pnueli^2_{(a, a+1)}(P, Q)$ that lies completely within $\tau_i + (a, a+1)$ (in which case $\pnueli^2_{(a, a+1)}(P, Q)$ holds),
or we are in the scenario in~\figurename~\ref{fig:scenario.one}---but this is impossible, as one of the disjuncts of $\varphi_{\text{out}}$ must hold at $i$, as argued above.
If $\pnueli^2_{(a, a+1)}(P, Q)$ holds at $i$, then we have a witness $\langle h, \ell \rangle$ for $\pnueli^2_{(a, a+1)}(P, Q)$ at $i$ that lies completely within $\tau_i + (a, a+1)$, and $\varphi_{\textit{wit}}$ clearly holds at $i$ too.
If $\textstyle{\counting_{\leq a}^k} (\A^2_{\geq 1}) \wedge \neg  \textstyle{\counting_{\leq a}^{k + 1}} (\A^2_{\geq 1})$ indeed holds
at $i$ for some $k$ then $\B^k_{\geq a + 1}$ must not hold at $i$: if $\langle h_j, \ell_j \rangle$ and $\langle h_{j+1}, \ell_{j+1} \rangle$ are potential witnesses for $\pnueli^2_{(a, a+1)}(P, Q)$ and $h_j$ is the $k$-th point satisfying $\varphi_2^{\geq 1}(x)$,
we must have $h_{j+1} \leq h$ and $\tau_{\ell_{j+1}} \in \tau_i + (a, a + 1)$.
\end{proof}

\begin{figure}[t]
\centering
\begin{minipage}{0.45\linewidth}
\centering
\scalebox{.70}{
\begin{tikzpicture}
\begin{scope}

\draw[-, very thick, loosely dotted] (-72pt,0pt) -- (-32pt,0pt);

\draw[-, very thick, loosely dotted] (160pt,0pt) -- (200pt,0pt);

\draw[loosely dashed] (-10pt,-30pt) -- (-10pt,10pt) node[at start, below=2mm] {$a$};

\draw[loosely dashed] (90pt,-30pt) -- (90pt,10pt) node[at start, below=2mm] {$a + 1$};

\draw[draw=black, fill=white] (1pt, -4pt) rectangle (-1pt, 4pt);

\draw[draw=black, fill=white] (53pt, -4pt) rectangle (51pt, 4pt);
\node[below, fill=white, inner sep=1mm] at (52pt, -7pt) {$\ell_{j}$};

\draw[draw=black, fill=white] (105pt, -4pt) rectangle (103pt, 4pt);

\draw[draw=black, fill=white] (151pt, -4pt) rectangle (149pt, 4pt);
\node[below, fill=white, inner sep=1mm] at (150pt, -7pt) {$\ell_{j+1}$};

\draw[draw=black, fill=white] (-21pt, -4pt) rectangle (-23pt, 4pt);
\node[below, fill=white, inner sep=1mm] at (-22pt, -7pt) {$h_{j}$};

\draw[draw=black, fill=white] (31pt, -4pt) rectangle (29pt, 4pt);

\draw[draw=black, fill=white] (80pt, -4pt) rectangle (78pt, 4pt);
\node[below, fill=white, inner sep=1mm] at (79pt, -7pt) {$h_{j+1}$};

\draw[draw=black, fill=white] (87pt, -4pt) rectangle (85pt, 4pt);

\draw[draw=black, fill=white] (120pt, -4pt) rectangle (118pt, 4pt);

\end{scope}

\begin{scope}
\draw[|<->|][dotted] (-22pt,25pt)  -- (150pt,25pt) node[midway,above] {{$\geq 1$}};
\end{scope}

\end{tikzpicture}
}
\caption{Undesired scenario \#1.}
\label{fig:scenario.one}
\end{minipage}
\begin{minipage}{0.45\linewidth}
\centering
\scalebox{.70}{
\begin{tikzpicture}
\begin{scope}

\draw[-, very thick, loosely dotted] (-92pt,0pt) -- (-52pt,0pt);

\draw[-, very thick, loosely dotted] (133pt,0pt) -- (173pt,0pt);

\draw[loosely dashed] (-10pt,-30pt) -- (-10pt,10pt) node[at start, below=2mm] {$a$};

\draw[loosely dashed] (90pt,-30pt) -- (90pt,10pt) node[at start, below=2mm] {$a + 1$};

\draw[draw=black, fill=white] (20pt, -4pt) rectangle (18pt, 4pt);
\node[below, fill=white, inner sep=1mm] at (19pt, -7pt) {$\ell_{j}$};

\draw[draw=black, fill=white] (10pt, -4pt) rectangle (8pt, 4pt);

\draw[draw=black, fill=white] (114pt, -4pt) rectangle (112pt, 4pt);

\node[below, fill=white, inner sep=1mm] at (81pt, -7pt) {$\ell_{j+1}$};

\draw[draw=black, fill=white] (-41pt, -4pt) rectangle (-43pt, 4pt);
\node[below, fill=white, inner sep=1mm] at (-42pt, -7pt) {$h_{j}$};

\draw[draw=black, fill=white] (46pt, -4pt) rectangle (44pt, 4pt);

\draw[draw=black, fill=white] (58pt, -4pt) rectangle (56pt, 4pt);
\node[below, fill=white, inner sep=1mm] at (57pt, -7pt) {$h_{j+1}$};

\draw[draw=black, fill=white] (82pt, -4pt) rectangle (80pt, 4pt);

\end{scope}

\begin{scope}
\draw[|<->|][dotted] (-42pt,25pt)  -- (81pt,25pt) node[midway,above] {{$\geq 1$}};
\end{scope}

\end{tikzpicture}
}
\caption{Desired scenario with $d(h_j, \ell_{j+1}) \geq 1$.}
\label{fig:scenario.one.desired}
\end{minipage}
\end{figure}

\section{Expressing more general properties in $\qtwomlo{}^\textsf{\textmd{fut}}$}\label{sec:substring}
We now consider the more general case where the desired 
behaviour in $I$ is 
specified as a future \fo{} formula $\psi(x', x'')$.\footnote{The proof applies also to the case where $\psi(x', x'')$ is a second-order formula.}
Formally, 
the property that we want to express is
\[
    \vartheta^\psi_{I}(x) = \exists x' \, \exists x'' \, \big(x < x' \leq x'' \wedge d(x, x') \in I \wedge d(x, x'') \in I \wedge \psi(x', x'') \big) \;.
\]
To simplify the analysis, we first modify $\psi(x', x'')$ into  $\psi_1(x', x'')$ to rule out witnesses that are not minimal: 
\[
\psi_1(x', x'') = \psi(x', x'') \wedge \neg \Big( \exists y \, \exists z \, \big(x' \leq y \leq z \leq x'' \wedge (x < y \vee z < x') \wedge \psi(y, z) \big) \Big) \;.
\]
Similarly as before, 
$\exists x'' \, \psi_1(x', x'')$
holds at the starting points of all the potential witnesses for $\vartheta^\psi_{I}$.
However, as opposed to the case of $P^2_{(a, a+1)}(P, Q)$, 
now the potential witnesses may overlap non-trivially.
In particular, if $\rho, i \models \psi_{\textit{wit}}$
where $\psi_{\textit{wit}}$
is defined in the same way as 
$\varphi_{\textit{wit}}$
in the last section, there is one more possible undesired scenario (illustrated in~\figurename~\ref{fig:scenario.two}; note in particular that $\psi_1(h_{j+2}, \ell_{j})$ does not hold).
Thanks to the finite-state nature of $\psi_1(x', x'')$,    
the scenario in~\figurename~\ref{fig:scenario.two} can also be ruled out in the same way: in this particular case, either $d(h_j, \ell_{j+1}) \geq 1$ or $d(h_{j+1}, \ell_{j+2}) \geq 1$ must hold. This is made possible by the following lemma that
gives an upper bound on the number of positions satisfying $\psi_2^{\geq 1}(x)$ (defined from $\psi_1(x', x'')$ in the same way as $\varphi_2^{\geq 1}(x)$) before $\tau_i + a$.\footnote{Similar observations based on Shelah's \emph{composition method}~\cite{Shelah1975} have also been used in~\cite{Hirshfeld06, HirshfeldR12}.}

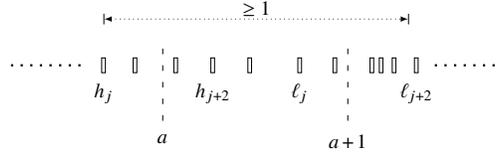
\begin{figure}[t]
\centering
\scalebox{.70}{
\begin{tikzpicture}
\begin{scope}

\draw[-, very thick, loosely dotted] (-92pt,0pt) -- (-52pt,0pt);

\draw[-, very thick, loosely dotted] (137pt,0pt) -- (177pt,0pt);

\draw[loosely dashed] (-10pt,-30pt) -- (-10pt,10pt) node[at start, below=2mm] {$a$};

\draw[loosely dashed] (90pt,-30pt) -- (90pt,10pt) node[at start, below=2mm] {$a + 1$};

\draw[draw=black, fill=white] (-2pt, -4pt) rectangle (-4pt, 4pt);

\draw[draw=black, fill=white] (18pt, -4pt) rectangle (16pt, 4pt);
\node[below, fill=white, inner sep=1mm] at (17pt, -7pt) {$h_{j+2}$};

\draw[draw=black, fill=white] (104pt, -4pt) rectangle (102pt, 4pt);

\draw[draw=black, fill=white] (109pt, -4pt) rectangle (107pt, 4pt);

\draw[draw=black, fill=white] (116pt, -4pt) rectangle (114pt, 4pt);

\draw[draw=black, fill=white] (128pt, -4pt) rectangle (126pt, 4pt);
\node[below, fill=white, inner sep=1mm] at (127pt, -7pt) {$\ell_{j+2}$};

\draw[draw=black, fill=white] (-41pt, -4pt) rectangle (-43pt, 4pt);
\node[below, fill=white, inner sep=1mm] at (-42pt, -7pt) {$h_{j}$};

\draw[draw=black, fill=white] (38pt, -4pt) rectangle (36pt, 4pt);

\draw[draw=black, fill=white] (65pt, -4pt) rectangle (63pt, 4pt);
\node[below, fill=white, inner sep=1mm] at (64pt, -7pt) {$\ell_{j}$};

\draw[draw=black, fill=white] (84pt, -4pt) rectangle (82pt, 4pt);

\draw[draw=black, fill=white] (-24pt, -4pt) rectangle (-26pt, 4pt);

\end{scope}

\begin{scope}
\draw[|<->|][dotted] (-42pt,25pt)  -- (123pt,25pt) node[midway,above] {{$\geq 1$}};
\end{scope}

\end{tikzpicture}
}
\caption{Undesired scenario \#2 where $\psi_1(h_{j+2}, \ell_{j})$ does not hold.}
\label{fig:scenario.two}
\end{figure}

\begin{lemma}
For any $\rho = (\sigma_i,\tau_i)_{i \geq 1}$ over $\Sigma_{\AP}$,
the \qtwomlo{} formula $\psi_2^{\geq 1}(x)$ over $\AP$ is satisfied by at most $(m + 1) \cdot (a + 1)$ positions $j > i$ (where $m$ is the number of locations in the minimal equivalent \dfa{} for $\psi_1(x, x')$) with $d(i, j) \in [0, a]$ for any $i \geq 0$.
\end{lemma}
\begin{proof}[Proof sketch.]
   Any point cannot intersect with more than $m$ potential witnesses for $\vartheta^\psi_{I}$ (otherwise there will be a contradiction with the minimality of potential witnesses), and this implies that 
   if $\rho, h_j \models \psi_2^{\geq 1}(x)$,
then either there is no $\langle h_{j+m+1}, \ell_{j+m+1} \rangle$ or $\tau_{h_{j+m+1}} \geq \tau_{h_{j}} + 1$.
\end{proof}
\noindent We then obtain the following theorem.
\begin{theorem}\label{thm:main}
The property `the future \foone{} formula $\psi(x', x'')$ is satisfied by positions $x'$, $x''$ in $I$ in the future' can be expressed in
$\emitl{}^\textsf{\textup{fut}} \equiv \qtwomlo{}^\textsf{\textup{fut}}$.
\end{theorem}
\noindent The theorem also holds for the general case where $\psi(x', x'')$ 
is a non-future \foone{} formula; in this case, the property can be expressed in $\emitl{} \equiv \qtwomlo{}$.

\section{Expressing $\counting_I^k$ in $\tlc{}^\textsf{\textmd{fut}}$}\label{sec:unilateral}

From~\cite{Ho19} we know that in the pointwise semantics, (aperiodic) \emitl{} (or \qtwomlo{}) formulae can be rewritten into simpler equivalent formulae where all intervals are unilateral, and in fact it suffices to use $[0, b\rangle$ and $[0, \infty)$~\cite{HoMadnani23}.
For the aperiodic case, such a formula can even be expressed with the simpler counting modalities as below, if we allow both the future and past versions of them:
\begin{itemize}
\item $\counting^k_I$ and $\pcounting^k_I$ with $I = (0, 1)$ in the continuous semantics~\cite{Hirshfeld06, HirshfeldR12};
or
\item $\counting^k_I$ and $\pcounting^k_I$ with $I = [0, b \rangle$ in the pointwise semantics~\cite{HoMadnani23}.
\end{itemize}
We now show that
for the special case of $\counting_I^k$, i.e.~when the $\qtwomlo{}^\textsf{fut}$ formula in question is a $\tlci{}^\textsf{fut}$ formula,  
we can do the same \emph{with only the future modalities}; this can be seen as a strict generalisation of the `well-known' reduction from $\until_I$ with $\langle a, b \rangle$ to $\until_I$ with $I = [0, b \rangle$ discussed earlier~\cite{HenRas98, HirshfeldR99, Hirshfeld06}.
In the presentation below we will focus on the pointwise case, where some additional conditions must be satisfied (as explained below), but these conditions are automatically satisfied in the continuous semantics. 

\paragraph{Expressing $\eventually_I$ with $I = \langle a, b \rangle$.}
We start by rewriting the `eventually' modalities $\eventually_I$, which can actually be regarded as a special case of $\counting_I^k$ with $k = 1$~\cite{Hirshfeld06}; for simplicity, let us consider a subformula $\eventually_I \varphi$ where $\varphi$ is in unilateral 
$\mitl{}^\textsf{fut}$ and $I = (a, a + 1)$, $a \geq 1$.
It is well known that in the pointwise semantics, such modalities cannot be expressed in  unilateral \mitl{}~\cite{Raskin1999}. To overcome this apparent difficulty,  
let us define a family of formulae for all $m \in \{0, \dots, a - 1 \}$:
\begin{IEEEeqnarray*}{rCl}
\Phi^0 & = & \{ \varphi \} \;, \\
\Phi^{m+1} & = &  \{ \nextx_{>0} \top \wedge \neg \varphi^m \until_{\leq 1} \varphi^m \wedge \neg \varphi^m \until_{\geq 1} \varphi^m, \globally_{(0, 1)} \varphi^m \, 
\mid \, \varphi^m \in \Phi^m \text{ or } \neg \varphi^m \in \Phi^m \} \;.
\end{IEEEeqnarray*}
All these formulae are in unilateral $\mitl{}^\textsf{fut}$: $\globally_{(0, 1)} \varphi^{m} \equiv (\nextx_{> 0} \top \wedge  \globally_{[0, 1)} \varphi^{m}) \vee \eventually_{\leq 0} \globally_{[0, 1)} \varphi^{m}$.
Additionally, we assume that the timed word
$\rho=(\sigma_i,\tau_i)_{i \geq 1}$ in question satisfies the following condition:
\begin{itemize}
    \item For every $m \in \{0, \dots, a - 1 \}$ and $\varphi^m \in \Phi^m$, if $\rho, j \models \varphi^m$ and
    $\rho, j' \centernot \models \varphi^m$ for all $j' < j$ with $d(j', j) < 1$, then there exists $i$ in $\rho$ such that $d(i, j) = 1$  (unless $d(1, j) < 1$).
\end{itemize}
We note that in practical applications, this should not be a severe limitation---for example in model checking, if the system is modelled as a timed automaton~\cite{AluDil94}, one can 
simply add a self-loop labelled with an extra `empty' letter $\epsilon$ to each location, and use the following formula (which is easily expressible in unilateral $\mitl{}^\textsf{fut}$) as a precondition:
\begin{IEEEeqnarray*}{rCll}
\vartheta^{\eventually} & = & \bigwedge_{\substack{\varphi^m \in \Phi^m \\ m \in \{0, \dots, a - 1\}}} \neg \exists x \, \exists x' \, \bigg( & x < x' \wedge \not \exists x'' \, (x < x'' < x') \\ & & & {} \wedge \exists y \, \Big(x < y \wedge d(x, y) > 1 \wedge d(x', y) < 1 \wedge \varphi^m(y) \wedge \not \exists z \, \big(x < z < y \wedge \varphi^m(z) \big) \Big) \bigg) \;.  
\end{IEEEeqnarray*}
Intuitively, $\vartheta^{\eventually}$ 
rules out the situations when $\nextx_{>0} \top \wedge \neg \varphi^m \until_{\leq 1} \varphi^m \wedge \neg \varphi^m \until_{\geq 1} \varphi^m \in \Phi^{m+1}$ should hold at  $x''$, but $x''$ does not exist in $\rho$.
With the condition in place, we now show that $\eventually_{\langle a - m, a-m+1 \rangle} \varphi^{m'}$ where $\varphi^{m'} \in \Phi^{m'}$ can be expressed in unilateral $\mitl{}^\textsf{fut}$ for $m \in \{0, \dots, a\}$ and $m' \leq m$.
For the base step $m = a$, note that $\eventually_{(0, 1\rangle} \varphi^{m'} \equiv (\nextx_{> 0} \top \wedge  \eventually_{[0, 1\rangle} \varphi^{m'}) \vee \eventually_{\leq 0} (\nextx_{> 0} \top \wedge  \eventually_{[0, 1\rangle} \varphi^{m'})$.
For the inductive step (from $m + 1$ to $m$), suppose that we want to express $\rho, i \models \eventually_{(a-m, a-m+1)} \varphi^m$ where $\varphi^m \in \Phi^m$ and let $\ell > i$ be the \emph{minimal} position such that $\rho, \ell \models \varphi^m$ and $d(i, \ell) \in (a-m, a-m + 1)$ (the arguments for other types of intervals are exactly similar). We can then essentially follow~\cite{Ho19} but only need to  consider the cases below:
\begin{itemize}
    \item There exists (a maximal) $j$, $i < j < \ell$ such that $d(j, \ell) = 1$ and $\rho, j \models \nextx_{>0} \top \wedge \neg \varphi^m \until_{\leq 1} \varphi^m \wedge \neg \varphi^m \until_{\geq 1} \varphi^m$: we have
\[
\rho, i \models \zeta_1 = \eventually_{(a-m-1, a-m)} (\nextx_{>0} \top \wedge \neg \varphi^m \until_{\leq 1} \varphi^m \wedge \neg \varphi^m \until_{\geq 1} \varphi^m)
\]
where $\nextx_{>0} \top \wedge \neg \varphi^m \until_{\leq 1} \varphi^m \wedge \neg \varphi^m \until_{\geq 1} \varphi^m \in \Phi^{m+1}$.

\item There exists $j$, $i < j < \ell$ such that $d(j, \ell)< 1$, $d(i, j) \in (a-m-1, a-m]$ and $\rho, j \models \varphi^{m}$: we have
\[
\rho, i \models \zeta_2 = \eventually_{(a-m-1, a-m]} \varphi^m \wedge 
\neg \eventually_{(a-m-1, a-m]} \globally_{(0, 1)} (\neg \varphi^m)
\]
where $\globally_{(0, 1)} (\neg \varphi^m) \in \Phi^{m+1}$.
\end{itemize}
The equivalent formula is $\zeta_1 \vee \zeta_2$, which can be rewritten into a unilateral $\mitl{}^\textsf{fut}$ formula by the induction hypothesis.
It follows that $\eventually_{\langle a, a+1 \rangle} \varphi^{0}$, where $\varphi^{0} = \varphi \in \Phi^{0}$ is an arbitrary $\mitl{}^\textsf{fut}$ formula, can be expressed in unilateral $\mitl{}^\textsf{fut}$, as desired.

\paragraph{Expressing $\counting^k_I$ with $I = \langle a, b \rangle$.}
We now consider a subformula $\counting_I^k \psi$ where $\psi$ is in $\tlc{}^\textsf{fut}$, $k \geq 2$, and $I = (a, a + 1)$, $a \geq 1$.
Define a family of formulae for all $m \in \{1, \dots, a - 1 \}$:
\begin{IEEEeqnarray*}{rCl}
\Psi^1 & = & \{ (\nextx_{>0} \top \vee \nextx_{\leq 0} \psi) \wedge \counting_{[0, 1]}^k \psi \wedge \neg \counting_{[0, 1)}^k \psi \} \;, \\
\Psi^{m+1} & = &  \{ \nextx_{>0} \top \wedge \neg \psi^m \until_{\leq 1} \psi^m \wedge \neg \psi^m \until_{\geq 1} \psi^m, \globally_{(0, 1)} \psi^m \,  \mid \, \psi^m \in \Psi^m \text{ or } \neg \psi^m \in \Psi^m \} \;.
\end{IEEEeqnarray*}
All these formulae are in $\tlc{}^\textsf{fut}$.
Now we assert that 
$\rho=(\sigma_i,\tau_i)_{i \geq 1}$  satisfies the following conditions:
\begin{enumerate}[label={(\bfseries C\arabic*)}]
    \item \label{item:cond1} If $\rho, j \models \psi$ and
    there are
    \begin{itemize}
    \item 
        less than $k$ positions
        $j' < j$ with $0 < d(j', j) < 1$ such that
        $\rho, j' \models \psi$, and
    \item 
        at least $k$ positions
        $j' \leq j$ with $0 \leq d(j', j) < 1$ such that
        $\rho, j' \models \psi$,
    \end{itemize}
    then there exists $i$ in $\rho$ such that $d(i, j) = 1$  (unless $d(1, j) < 1$).
    
    \item \label{item:cond2} For every $m \in \{1, \dots, a - 1 \}$ and $\psi^m \in \Psi^m$, if $\rho, j \models \psi^m$ and
    $\rho, j' \centernot \models \psi^m$ for all $j' < j$ with $d(j', j) < 1$, then there exists $i$ in $\rho$ such that $d(i, j) = 1$  (unless $d(1, j) < 1$).
\end{enumerate}
As before, we can use $\vartheta^{\eventually}$ (trivially modified so that the conjunction ranges over $m \in \{1, \dots, a - 1 \}$) to enforce the second condition.
For the first condition we assert the formula
\[
\varphi^{\counting} = \neg \Big(
\weventually \big(\nextx_{> 0} (\neg \psi) \wedge \neg \counting_{[0, 1]}^k \psi \wedge \nextx \counting_{[0, 1)}^k \psi\big)
\vee
\weventually (\nextx_{> 0} \psi \wedge \neg \counting_{[0, 1]}^{k} \psi \wedge \nextx \counting_{[0, 1)}^{k-1} \psi)
\Big) \;.
\]
\begin{lemma}\label{lem:countingcondition}
$\rho, 1 \models \varphi^{\counting}$ iff the first condition above holds.     
\end{lemma}
\begin{proof}
Assume that the first condition is violated and there are two adjacent positions $x, x' < j$ such that $d(x, j) > 1$ and $d(x', j) < 1$. Consider the following cases:
\begin{itemize}
    \item $\rho, x' \centernot \models \psi$:
    It is clear that $\rho, x' \models \counting_{[0, 1)}^k \psi$, since the covered period may contain positions $j' > j$ with $d(j, j') > 0$, and excluding $x'$ makes no difference.
    It is also clear that $\rho, x \centernot \models \counting_{[0, 1]}^k \psi$ as the covered period may only contain fewer positions. We thus have $\rho, i \centernot \models \varphi^{\counting}$.
    
    \item $\rho, x' \models \psi$:
    It is clear that $\rho, x' \models \counting_{[0, 1)}^{k-1} \psi$ as the covered period must contain at least $k-1$ positions satisfying $\psi$ after excluding $x'$.
    It is also clear that $\rho, x \centernot \models \counting_{[0, 1]}^k \psi$ as the covered period may only contain fewer positions. We thus have $\rho, i \centernot \models \varphi^{\counting}$.
\end{itemize}
For the other direction, consider the following cases:
\begin{itemize}
    \item $\rho, x \models \nextx_{> 0} (\neg \psi) \wedge \neg \counting_{[0, 1]}^k \psi \wedge \nextx \counting_{[0, 1)}^k \psi$ for some position $x$: Let the next position be $x'$.
    It is clear that there is at least one position satisfying $\psi$ in $(\tau_x + 1, \tau_{x'} + 1)$.
    Let $j$ be the position such that $|\{ j' \, \mid \, x < j' \leq j \text{ and } \rho, j' \models \psi \}| = k$. It is clear that $j$ satisfies the statements in the condition, but by assumption, there is no $i$ in $\rho$ such that $d(i, j) = 1$.
    
    \item $\rho, x \models \nextx_{> 0} \psi \wedge \neg \counting_{[0, 1]}^{k} \psi \wedge \nextx \counting_{[0, 1)}^{k-1} \psi$ for some position $x$: Let the next position be $x'$.
    Once again it is clear that there is at least one position satisfying $\psi$ in $(\tau_x + 1, \tau_{x'} + 1)$. The argument is identical to the previous case. \qedhere
\end{itemize}
\end{proof}
\noindent 
We say that a segment $\langle h, \ell \rangle$ is a witness for $\counting_{(a, a+1)}^k \psi$ at $i$ if $h < \ell$, $\rho, h \models \psi$, $\rho, \ell \models \psi$, $|\{ j \, \mid \, h \leq j \leq \ell \text{ and } \rho, j \models \psi \}| = k$,
and both $\tau_h, \tau_\ell \in \tau_i + (a, a+1)$.
Similarly as before, we can write an untimed (finite-word) $\ltl{}^\textsf{fut}$ 
formula $\psi_1$ that holds at all the starting points $h$ of all the potential witnesses (ignoring the timing requirement) for $\counting_{(a, a+1)}^k \psi$---in this case, it is simply an untimed (finite-word) \ltl{} formula that counts exactly $k$ occurrences of $\psi$.
Based on this, we can give an initial attempt to express 
$\counting_{(a, a+1)}^k \psi$, similar to what we did for $\pnueli^2_{(a, a+1)}(P, Q)$ 
using $\qtwomlo{}^\textsf{fut}$
in Section~\ref{sec:pnueli}:
\begin{IEEEeqnarray*}{rCll}
\varphi_{\textit{wit}} & = & 
\eventually_{(a, a+1)} \psi_1 \wedge
 \Big( & \eventually_{(a-1, a)} \big((\nextx_{>0} \top \vee \nextx_{\leq 0} \psi) \wedge \counting_{[0, 1]}^k \psi \wedge \neg \counting_{[0, 1)}^k \psi\big) \\
 & & & {} \vee \big( \eventually_{(a-1, a]} \psi \wedge \globally_{(a-1, a]} (\psi \implies \counting_{[0, 1)}^k \psi) \big) \Big)  \;.
\end{IEEEeqnarray*}
Here, however, the second conjunct is more involved as we must refrain from using (aperiodic) automata modalities. We now prove some propositions about the correctness of $\varphi_{\textit{wit}}$,
based on the assumption that $\rho$
satisfies~\ref{item:cond1} and ~\ref{item:cond2}.

\begin{proposition}
$\rho, i \models \eventually_{(a-1, a]} \psi \wedge \globally_{(a-1, a]} (\psi \implies \counting_{[0, 1)}^k \psi)$ implies $\rho, i \models \counting_{(a, a+1)}^k \psi$.
\end{proposition}
\begin{proof}
Let $j$ be the \emph{maximal} position in $\rho$ such that $d(i, j) \in (a-1, a]$ and $\rho, j \models \psi$. We have $\rho, j \models \counting_{[0, 1)}^k \psi$, and it is clear that $\tau_i + (a, a+1)$ contains at least $k$ positions satisfying $\psi$.
\end{proof}

\begin{proposition}\label{prop:counttowit1}
$\rho, i \models \counting_{(a, a+1)}^k \psi \wedge \neg \eventually_{(a-1, a]} \psi$ implies that 
$\rho, i \models \eventually_{(a-1, a)} \big((\nextx_{>0} \top \vee \nextx_{\leq 0} \psi) \wedge \counting_{[0, 1]}^k \psi \wedge \neg \counting_{[0, 1)}^k \psi\big)$.
\end{proposition}
\begin{proof}
 Let $j_1$ be the minimal position in $\rho$ such that $d(i, j_1) \in (a, a+1)$ and $\rho, j_1 \models \psi$, and $j_k$ be the position in $\rho$ such that $\sigma_{j_1} \dots \sigma_{j_k} \models \psi_1$ and $\rho, j_k \models \psi$. By Lemma~\ref{lem:countingcondition},
 the first condition above holds and there is a position $j'$ such that $d(j', j_{k}) = 1$ and
 $\rho, j' \models (\nextx_{>0} \top \vee \nextx_{\leq 0} \psi) \wedge \counting_{[0, 1]}^k \psi \wedge \neg \counting_{[0, 1)}^k \psi$.
\end{proof}

\begin{proposition}\label{prop:counttowit2}
$\rho, i \models \counting_{(a, a+1)}^k \psi \wedge  \eventually_{(a-1, a]} \psi$ implies that
$\rho, i \models \eventually_{(a-1, a)} \big((\nextx_{>0} \top \vee \nextx_{\leq 0} \psi) \wedge \counting_{[0, 1]}^k \psi \wedge \neg \counting_{[0, 1)}^k \psi\big)$ or $\rho, i \models \eventually_{(a-1, a]} \psi \wedge \globally_{(a-1, a]} (\psi \implies \counting_{[0, 1)}^k \psi)$. 
\end{proposition}
\begin{proof}
 Let $j_1$ be the minimal position in $\rho$ such that $d(i, j_1) \in (a, a+1)$ and $\rho, j_1 \models \psi$, and $j_k$ is the minimal position in $\rho$ such that there exists $j_1 < \dots < j_k$ where $\rho, j_i \models \psi$ for all $i \in \{1, \dots, k\}$.
Let $\ell$ be the \emph{maximal} position in $\rho$ such that $d(i, \ell) \in (a-1, a]$ and $\rho, \ell \models \psi$. Consider the following cases:
\begin{itemize}
\item $d(\ell, j_k) \geq 1$: By Lemma~\ref{lem:countingcondition}, there exists a position $\ell' \geq \ell$ in $\rho$ such that $d(\ell', j_k) = 1$ and $\rho, \ell' \models (\nextx_{>0} \top \vee \nextx_{\leq 0} \psi) \wedge \counting_{[0, 1]}^k \psi \wedge \neg \counting_{[0, 1)}^k \psi$.
\item $d(\ell, j_k) < 1$: We have $\rho, \ell \models \counting_{[0, 1)}^k \psi$. Now consider $j_{k-1}$ and the largest position $\ell' < \ell$ such that $d(i, \ell') \in (a-1, a]$ and $\rho, \ell' \models \psi$. If $d(\ell', j_{k - 1}) < 1$ then clearly 
$\rho, \ell' \models \counting_{[0, 1)}^k \psi$.
If $d(\ell', j_{k - 1}) \geq 1$, then by 
Lemma~\ref{lem:countingcondition},
 there exists $\ell'' \geq \ell'$ such that $d(i, \ell'') \in (a-1, a]$, $d(\ell'', j_{k-1}) = 1$, and $\rho, \ell'' \models (\nextx_{>0} \top \vee \nextx_{\leq 0} \psi) \wedge \counting_{[0, 1]}^k \psi \wedge \neg \counting_{[0, 1)}^k \psi$. The argument is repeated until some position in $\tau_i + (a-1, a]$ satisfies $(\nextx_{>0} \top \vee \nextx_{\leq 0} \psi) \wedge \counting_{[0, 1]}^k \psi \wedge \neg \counting_{[0, 1)}^k \psi$, or all positions in $\tau_i + (a-1, a]$ satisfying $\psi$ also satisfy $\counting_{[0, 1)}^k \psi$. \qedhere
\end{itemize}
\end{proof}
It remains to strengthen
$\eventually_{(a, a+1)} \psi_1 \wedge
 \eventually_{(a-1, a)} \big((\nextx_{>0} \top \vee \nextx_{\leq 0} \psi) \wedge \counting_{[0, 1]}^k \psi \wedge \neg \counting_{[0, 1)}^k \psi\big)$ so that it implies $\counting_{(a, a+1)}^k \psi$. As before in Section~\ref{sec:pnueli}, we need a formula $\psi_2^{k}$ that refers to two neighbouring potential witnesses---in this case, it is simply an untimed (finite-word) $\ltl{}^\textsf{fut}$ formula that counts exactly $k + 1$ occurrences of $\psi$.
We can then argue that there is an upper bound on the number of positions satisfying $\psi_2^{k, \geq 1}$ 
(easily expressible in $\tlc{}^\textsf{fut}$) before $\tau_i + a$.
In contrast to Section~\ref{sec:pnueli}, however, we need an alternative way to express $\B^k_{\geq a + 1}$.
\begin{lemma}\label{lem:psikgeq1occurrence}
For any $\rho = (\sigma_i,\tau_i)_{i \geq 1}$ over $\Sigma_{\AP}$,
the $\tlc{}^\textsf{fut}$ formula $\psi_2^{k, \geq 1}$ over $\AP$ is satisfied by at most $k \cdot a$ positions $j > i$ with $d(i, j) \in [0, a)$ for any $i \geq 0$.
\end{lemma}
\begin{lemma}
For any $\rho = (\sigma_i,\tau_i)_{i \geq 1}$ over $\Sigma_{\AP}$,
the $\tlc{}^\textsf{fut}$ formula $(\nextx_{>0} \top \vee \nextx_{\leq 0} \psi) \wedge \counting_{[0, 1]}^k \psi \wedge \neg \counting_{[0, 1)}^k \psi$ over $\AP$ is satisfied by at most $k \cdot (a + 1) + 1$ positions $j > i$ with $d(i, j) \in [0, a]$ for any $i \geq 0$.
\end{lemma}
\begin{proof}[Proof sketch.]
Each occurrence $j$ of $(\nextx_{>0} \top \vee \nextx_{\leq 0} \psi) \wedge \counting_{[0, 1]}^k \psi \wedge \neg \counting_{[0, 1)}^k \psi$, except for possibly the first one, happens `between' two neighbouring (minimal) potential witnesses $\langle h_i, \ell_i \rangle$ and $\langle h_{i+1}, \ell_{i+1} \rangle$ with $d(h_i, \ell_{i+1}) \geq 1$: either $j = h_i$ or $h_i < j < h_{i+1}$. The claim holds by (a trivial modification of) Lemma~\ref{lem:psikgeq1occurrence}. 
\end{proof}
Now suppose that $\rho, i \models \eventually_{(a, a+1)} \psi_1 \wedge
 \eventually_{(a-1, a)} \big((\nextx_{>0} \top \vee \nextx_{\leq 0} \psi) \wedge \counting_{[0, 1]}^k \psi \wedge \neg \counting_{[0, 1)}^k \psi\big)$
 and let $j$ be the maximal position in $\tau_i + (a-1, a)$ such that $\rho, j \models (\nextx_{>0} \top \vee \nextx_{\leq 0} \psi) \wedge \counting_{[0, 1]}^k \psi \wedge \neg \counting_{[0, 1)}^k \psi$.
The undesired scenario is when
there is a maximal $j' > j$ with $\tau_{j'} \in \tau_i + (a-1, a]$ such that $\rho, j' \models \psi$, and there are less than $k$ positions in $\tau_i + (a, a+1)$ satisfying $\psi$. In this case, it is clear that $\rho, j' \models \psi_2^{k, \geq 1}$. 
To rule this scenario out we employ the following strategy, which can be implemented as a $\tlc{}^\textsf{fut}$ formula (which we opt to explain in English, for the sake of readability; we count events at positions $> i$):
\begin{enumerate}[label={(\arabic*)}]
\item \label{item:undesired1} Count the number of occurrences of $\psi_2^{k, \geq 1}$ in $\tau_i + [0, a)$  and $\tau_i + [0, a]$.
   If they do not match, then we are in the undesired scenario. 
\item \label{item:undesired2} Count the number of occurrences of $(\nextx_{>0} \top \vee \nextx_{\leq 0} \psi) \wedge \counting_{[0, 1]}^k \psi \wedge \neg \counting_{[0, 1)}^k \psi$ in $\tau_i + [0, a)$.
\item \label{item:undesired3} Take a disjunction over all the possible ways in which these occurrences may interleave in $\tau_i + [0, a)$ (note that they may hold simultaneously on the same position). Those ending with $\psi_2^{k, \geq 1}$ but not $(\nextx_{>0} \top \vee \nextx_{\leq 0} \psi) \wedge \counting_{[0, 1]}^k \psi \wedge \neg \counting_{[0, 1)}^k \psi$
are in the undesired scenario.
\end{enumerate}
Let us call this formula (which captures the undesired scenario) $\varphi'_{\textit{out}}$.

\begin{proposition}
$\eventually_{(a, a+1)} \psi_1 \wedge
 \eventually_{(a-1, a)} \big((\nextx_{>0} \top \vee \nextx_{\leq 0} \psi) \wedge \counting_{[0, 1]}^k \psi \wedge \neg \counting_{[0, 1)}^k \psi\big) \wedge \neg \varphi'_{\textit{out}}$ implies $\rho, i \models \counting_{(a, a+1)}^k \psi$.
\end{proposition}
\begin{proof}
Consider the conditions above that form $\varphi'_{\textit{out}}$.
First note that if the number of occurrences of $\psi_2^{k, \geq 1}$ in $\tau_i + [0, a]$ is $0$, then it is easy to see that $\rho, i \models \counting_{(a, a+1)}^k \psi$.
To see~\ref{item:undesired1}, note that  
if $\psi_2^{k, \geq 1}$ holds at some position $j'$ with $\tau_{j'} = \tau_i + a$, then $\tau_i + (a, a +1)$ may contain at most $k - 1$ positions satisfying $\psi$.
So for~\ref{item:undesired3}, first
assume that $\psi_2^{k, \geq 1}$ does
not hold at any position at $\tau_i + a$.
Let $j$ be the maximal position with $\tau_j \in \tau_i + (a-1, a)$ such that $\rho, j \models (\nextx_{>0} \top \vee \nextx_{\leq 0} \psi) \wedge \counting_{[0, 1]}^k \psi \wedge \neg \counting_{[0, 1)}^k \psi$
and $j' > i$ be the maximal position with
$\tau_{j'} \in \tau_i + [0, a)$ and $\rho, j' \models \psi_2^{k, \geq 1}$.

If $\psi$ holds at some maximal position $\ell$ at $\tau_i + a$, since $\rho, i \models \eventually_{(a, a+1)} \psi_1$, we have
$\rho, \ell \models \psi_2^{k, < 1}$ (defined in the expected way) and thus $\rho, i \models \counting_{(a, a+1)}^k \psi$;
we argue that 
$\psi_2^{k, \geq 1}$ cannot hold at any $\ell'$ where $j < \ell' < \ell$. Suppose to the contrary that $\rho, \ell' \models \psi_2^{k, \geq 1}$
(Wlog. let $\ell'$ be the largest such position
at the same timestamp $\tau_{\ell'}$).
If $\rho, \ell' \models \psi_2^{k, = 1}$, we have  
$\rho, \ell' \models (\nextx_{>0} \top \vee \nextx_{\leq 0} \psi) \wedge \counting_{[0, 1]}^k \psi \wedge \neg \counting_{[0, 1)}^k \psi$, contradicting the maximality of $j$.
If $\rho, \ell' \models \psi_2^{k, > 1}$ then
by Lemma~\ref{lem:countingcondition},
there exists a position $\ell'' > \ell'$ such that $\rho, \ell'' \models (\nextx_{>0} \top \vee \nextx_{\leq 0} \psi) \wedge \counting_{[0, 1]}^k \psi \wedge \neg \counting_{[0, 1)}^k \psi$, again contradicting the maximality of $j$.
Now we assume that $\psi$ does not hold at any position at $\tau_i + a$.
Consider the following cases:
\begin{itemize}
    \item $j' = j$: We know that $\rho, j \models \psi$. Consider the following subcases:
    \begin{itemize}
        \item $j'$ is not the maximal position with $\tau_{j'} \in \tau_i + (a-1, a)$ such that $\rho, j' \models \psi$: There is a maximal $j''$ with $\tau_{j''} \in \tau_i + (a-1, a)$ and $\rho, j'' \models \psi$. Since $\rho, i \models \eventually_{(a, a+1)} \psi_1$
        and thus $\rho, j'' \models \psi_2^{k, < 1}$, it follows that $\rho, i \models \counting_{(a, a+1)}^k \psi$.
        \item  $j'$ is the maximal position with $\tau_{j'} \in \tau_i + (a-1, a)$ such that $\rho, j' \models \psi$: Since $\rho, j' \models \counting_{[0, 1]}^k$ but there is no $j'' > j$ with $\tau_{j''} \in \tau_i + (a-1, a]$ such that $\rho, j'' \models \psi$, we have $\rho, i \models \counting_{(a, a+1)}^k \psi$.
    \end{itemize}
    \item $j' < j$: If there is a maximal $j'' > j$ with $\tau_{j''} \in \tau_i + (a-1, a)$ and $\rho, j'' \models \psi$,
    Since $\rho, i \models \eventually_{(a, a+1)} \psi_1$
        we have $\rho, j'' \models \psi_2^{k, < 1}$, and it follows that $\rho, i \models \counting_{(a, a+1)}^k \psi$.
        If there is no such $j''$, since $\rho, j \models \counting_{[0, 1]}^k$ we also have $\rho, i \models \counting_{(a, a+1)}^k \psi$.
    \qedhere
\end{itemize}
\end{proof}
Our final formula for $\counting_{(a, a+1)}^k \psi$ is 
\begin{IEEEeqnarray*}{rCll}
\varphi'_{\textit{wit}} & = & 
\eventually_{(a, a+1)} \psi_1 \wedge
 \bigg( & \Big( \eventually_{(a-1, a)} \big((\nextx_{>0} \top \vee \nextx_{\leq 0} \psi) \wedge \counting_{[0, 1]}^k \psi \wedge \neg \counting_{[0, 1)}^k \psi\big) \wedge \neg \varphi'_{\textit{out}} \Big) \\
 & & & {} \vee \big( \eventually_{(a-1, a]} \psi \wedge \globally_{(a-1, a]} (\psi \implies \counting_{[0, 1)}^k \psi) \big) \bigg)  \;.
\end{IEEEeqnarray*}
To see that $\counting_{(a, a+1)}^k \psi$ implies $\varphi_{\textit{wit}}'$, observe that Propositions~\ref{prop:counttowit1} and~\ref{prop:counttowit2} still hold if 
the conjunct $\neg \varphi_{\textit{wit}}'$ is added.
We apply the equivalence repeatedly from the innermost subformula $\counting_I^k \psi$ where $\psi$ is in $\tlc{}^\textsf{fut}$, and then work outwards
until there is no $\counting_I^k \psi$ with $I = \langle a, b \rangle$. In the process, we also need to ensure that~\ref{item:cond1} and ~\ref{item:cond2} are satisfied for various $\psi$.
Finally, we rewrite $\eventually_I$ with $I = \langle a, b \rangle$ into $\eventually_I$ with $I = [0, b \rangle$.
\begin{theorem}
Given a $\tlci{}^\textsf{fut}$ formula $\varphi$,
there is a $\tlc{}^\textsf{fut}$ formula $\varphi'$
such that $\varphi$ and $\varphi'$ are equivalent over timed words satisfying~\emph{\ref{item:cond1}} and ~\emph{\ref{item:cond2}} (for some finite set of $\psi$). 
\end{theorem}
Finally, this result carries over to the case of the continuous interpretations of \tlci{}, as the `positions'
postulated by~\ref{item:cond1} and~\ref{item:cond2} automatically exist.
\begin{corollary}
$\tlci{}^\textsf{\textup{fut}} \subseteq \tlc{}^\textsf{\textup{fut}}$ in the continuous semantics.    
\end{corollary}

\section{Conclusion and future work}
It turned out that allowing $\langle a, b \rangle$ in counting modalities 
only makes them more intricate to express in (aperiodic) automata modalities (or \qtwomlo{}), which necessarily `start' from the current point; in other words, the relevant claims in~\cite{Rabinovich2010} are indeed correct. 
More generally, we have shown that the existence of two `witnesses' $x' \leq x''$ 
for a first-order formula $\varphi(x', x'')$
in $t_0 + \langle a, b \rangle$ 
can also be captured in aperiodic $\emitl{}^\textsf{fut}$ (or $\qtwomlo{}^\textsf{fut}$).
This is somewhat surprising, as the timing constraints on both $x'$ and $x''$
does seem to require the use of 
punctualities or non-trivial  extensions.
Our second main result gives a satisfactory correction
to the folklore belief, at least in the case of continuous semantics. We list below some possible further directions:

\begin{itemize}


\item \mitl{} with both the future and past modalities and rational constants appears to be very expressive with $\mathrm{EXPSPACE}$-complete satisfiability and model-checking problems (through a simple scaling argument). We also know from Theorem~\ref{thm:rational} and~\cite{Hunter2013} that it can be made expressively complete for \foone{} by adding punctualities in the continuous semantics. Can it express some decidable fragments of \onetptl{}~\cite{KrishnaMMP22} with rational constants (i.e.~without using automata modalities)?

\item The properties considered in Section~\ref{sec:substring} can be seen as a special case of a decidable fragment of the logic \pnemtl{} recently proposed in~\cite{KrishnaMMP22}. Can we extend the ideas presented here to handle more general \pnemtl{} properties, where automata modalities do not start from the current point? 

\item Can the construction in Section~\ref{sec:unilateral}
lead to a future (or `almost future'~\cite{590a71927dcb4ba78bc7b58a4d9bfbe8}) metric temporal logic that is expressively complete for $\qtwomlo{}^\textsf{fut}$, or more generally a separation result akin to~\cite{Gabbay1980} or~\cite{HunterOW13}?

\end{itemize}

\bibliographystyle{eptcs}
\bibliography{refs}
\end{document}